\documentclass[runningheads]{llncs}

\usepackage{amsmath}
\usepackage{amssymb}
\usepackage{graphicx}
\usepackage{subfigure}
\usepackage{algorithm}
\usepackage{algorithmic}

\renewcommand{\algorithmicrequire}{\textbf{Input:}}
\renewcommand{\algorithmicensure}{\textbf{Output:}}
\renewcommand{\algorithmicprint}{\textbf{execute}}

\newcommand{\removed}[1]{}

\newcommand{\co}{\mathcal{O}}

\title{Passively Mobile Communicating Machines that Use Restricted Space\thanks{This work has been partially supported by the ICT Programme of the European Union under contract number ICT-2008-215270 (\textsf{FRONTS}).}}

\author{Ioannis Chatzigiannakis\inst{1,2} \and Othon Michail\inst{1,2} \and Stavros Nikolaou\inst{2} \and \\ Andreas Pavlogiannis\inst{2} \and Paul G. Spirakis\inst{1,2}}
\institute{Research Academic Computer Technology Institute (RACTI), Patras, Greece\and Computer Engineering and Informatics Department (CEID), University of Patras, Patras, Greece.\\
Email:\email{ \{ichatz, michailo, spirakis\}@cti.gr, \\ snikolaou@ceid.upatras.gr, paulogiann@ceid.upatras.gr}}

\titlerunning{Passively Mobile Communicating Machines that Use Restricted Space}

\begin{document}

\maketitle

\begin{abstract}
We propose a new theoretical model for passively mobile Wireless Sensor Networks, called $PM$ (standing for \emph{Passively mobile Machines}). The main modification w.r.t. the Population Protocol model \cite{AADFP06} is that agents now, instead of being automata, are Turing Machines. We provide general definitions for unbounded memories, but we are mainly interested in computations upper-bounded by $\log n$, with $n$ being the size of the population (a space bound that clearly keeps our devices tiny). However, our results easily generalize to any space bound which is at least $\log n$. We focus on \emph{complete communication graphs} and define the complexity class $PMSPACE(f(n))$ consisting of all predicates on input assignments that are stably computable by some PM protocol that uses $\mathcal{O}(f(n))$ memory. We assume that the agents are initially \emph{identical} and have \emph{no global knowledge} of the system, which together with the completeness of the communication graph implies that all computable predicates have to be \emph{symmetric}. We show that there exists a PM protocol that \emph{can assign unique consecutive ids to the agents and inform them of the population size}, only by using $\mathcal{O}(\log n)$ memory. This allows us to give a direct simulation of any \emph{Deterministic} Turing Machine of $\mathcal{O}(n\log n)$ space by PM protocols using $\mathcal{O}(\log n)$ space. Going one step further, we generalize the simulation of the deterministic TM to prove that the PM model can simulate any \emph{Nondeterministic} TM of $\mathcal{O}(n\log n)$ space under the same space bounds. Moreover, by showing that a Nondeterministic TM of $\mathcal{O}(n\log n)$ space decides any language in $PMSPACE(\log n)$, we end up with an exact characterization for $PMSPACE(\log n)$: \emph{it is precisely the class of all symmetric predicates in} $NSPACE(n\log n)$. All these results easily generalize to the following exact characterizations for the classes $PMSPACE(f(n))$, for all $f(n)=\Omega(\log n)$: \emph{they are precisely the classes of all symmetric predicates in} $NSPACE(nf(n))$. In this way, we provide a space hierarchy for the PM model when the memory bounds are $\Omega(\log n)$. Finally, we explore the computability of the PM model when the protocols use $o(\log\log n)$ space and prove that $SEMILINEAR=PMSPACE(f(n))$ when $f(n)=o(\log\log n)$, where $SEMILINEAR$ denotes the class of the semilinear predicates. In fact, we prove that this bound is tight, so that $SEMILINEAR \subsetneq PMSPACE(f(n))$ when $f(n)=\mathcal{O}(\log \log n)$. 
\end{abstract}

\section{Introduction}

Theoretical models for Wireless Sensor Networks have received great attention over the past few years. Recently, Angluin \emph{et al.} \cite{AADFP06} proposed the \emph{Population Protocol} (\emph{PP}) model. Their aim was to model sensor networks consisting of tiny computational devices with sensing capabilities that follow some unpredictable and uncontrollable mobility pattern. Due to the minimalistic nature of their model, the class of computable predicates was proven to be fairly small: it is the class of \emph{semilinear predicates}, thus, not including multiplication of variables, exponentiations, and many other important operations on input variables. Moreover, Delporte-Gallet \emph{et al.} \cite{DFGR06} showed that PPs can tolerate only $\mathcal{O}(1)$ crash failures and not even a single Byzantine agent.

The work of Angluin \emph{et al.} shed light and opened the way towards a brand new and very promising direction. The lack of control over the interaction pattern, as well as its inherent nondeterminism, gave rise to a variety of new theoretical models for WSNs. Those models draw most of their beauty precisely from their inability to organize interactions in a convenient and predetermined way. In fact, the Population Protocol model was the minimalistic starting-point of this area of research. Most efforts are now towards strengthening the model of Angluin \emph{et al.} with extra realistic and implementable assumptions, in order to gain more computational power and/or speed-up the time to convergence and/or improve fault-tolerance \cite{CMS09-2}, \cite{GR09}.

In this work, we think of each agent as being a Turing Machine. In particular, we propose a new theoretical model for passively mobile sensor networks, called the \emph{PM} model. It is a model of Passively mobile Machines (that we keep calling agents) with sensing capabilities, equipped with two-way communication. We focus on PM protocols that use $\mathcal{O}(\log n)$ memory since, having \emph{logarithmic communicating machines} seems to be more natural than communicating automata of constant memory. First of all, the \emph{communicating machines} assumption is perfectly consistent with current technology (cellphones, iPods, PDAs, and so on). Moreover, \emph{logarithmic} is, in fact, \emph{extremely small}. In addition, we explore the computability of the PM model on different space bounds in order to get an insight of the trade-off between computational power and resource (memory) availability. As will shall see, in PM protocols that use $f(n) = \Omega(\log n)$ space, agents can be organized into a distributed nondeterministic TM that makes use of all the available space. In the case, where $f(n) = o(\log\log n)$ however, we show that the PM protocols are computationally equal to Population Protocols. Thus, we provide exact characterizations for the input symmetric computations performed by communicating TMs using the above space bounds.


\subsection{Other Previous Work}

In \cite{AADFP06}, the \emph{Probabilistic Population Protocol} model was proposed, in which the scheduler selects randomly and uniformly the next pair to interact. Some recent work has concentrated on performance, supported by this random scheduling assumption (see e.g. \cite{AAE08}). \cite{CDFMS09} proposed a generic definition of probabilistic schedulers and a collection of new fair schedulers, and revealed the need for the protocols to adapt when natural modifications of the mobility pattern occur. \cite{BCCGK09,CS08} considered a huge population hypothesis (population going to infinity), and studied the dynamics, stability and computational power of probabilistic population protocols by exploiting the tools of continuous nonlinear dynamics.

In addition, several extensions of the basic model have been proposed in order to more accurately reflect the requirements of practical systems. The \emph{Mediated Population Protocol} (\emph{MPP}) model of \cite{CMS09-3} was based on the assumption that each edge of the communication graph can store a state. It has been recently proven \cite{CMNPS10} that in the case of complete graphs the corresponding class is the symmetric subclass of $NSPACE(n^2)$, rendering the MPP model extremely powerful. Guerraoui and Ruppert \cite{GR09} made another natural assumption: each agent has its own unique id and can store up to a constant number of other agents' ids. In this model, which they named \emph{Community Protocol} model, the only permitted operation on ids is comparison. It was proven that the corresponding class consists of all symmetric predicates in $NSPACE(n\log n)$. In \cite{AACFJP05}, Angluin \emph{et al.} studied what properties of restricted communication graphs are stably computable, gave protocols for some of them, and proposed an extension of the model with \textit{stabilizing inputs}. In \cite{CMS10-2}, MPP's ability to decide graph properties was studied and it was proven that connectivity is undecidable. Some other works incorporated agent failures \cite{DFGR06} and gave to some agents slightly increased computational power \cite{BCMRR07} (heterogeneous systems). Recently, Bournez \emph{et al.} \cite{BCCK08} investigated the possibility of studying population protocols via game-theoretic approaches. For an excellent introduction to the subject of population protocols see \cite{AR07} and for some recent advances mainly concerning mediated population protocols see \cite{CMS09-4}.

\section{Our Results - Roadmap}

In Section \ref{sec:mod}, we begin with a formal definition of the PM model. The section proceeds with a thorough description of the systems' functionality and then provides definitions of \emph{configurations} and \emph{fair executions}. In Section \ref{sec:pred}, first \emph{stable computation} and then the complexity classes $SSPACE(f(n))$, $SNSPACE(f(n))$ (symmetric predicates in $SPACE(f(n))$ and $NSPACE(f($ $n))$, respectively), and $PMSPACE(f(n))$ (stably computable predicates by the PM model using $\mathcal{O}(f(n))$ space) are defined, while in section \ref{exam} we give two examples of PM protocols. In Section \ref{sec:uids}, we prove that PM protocols can assume the existence of unique consecutive ids and knowledge of the population size at the space cost of $\mathcal{O}(\log n)$ (Theorem \ref{the:iplm}). In section \ref{plm} we show how to exploit this knowledge in order for the system to simulate a non-deterministic TM (Theorem \ref{the:lowPLM}). This, along with Theorem \ref{the:plmup} provide an exact characterization of the PLM: $PLM=SNSPACE(nf(n))$ (Theorem \ref{plm_exact}). Based on the results of this section, we establish a space hierarchy theorem for the PM model, when the corresponding protocols use $\mathcal{O}(\log n)$ space (Theorem \ref{the:pmsh}). In section \ref{threshold} we examine the interesting case of the $o(\log log n)$ space bounded protocols, showing that it acts as a computability threshold of some sort. Finally, in Section \ref{sec:conc} we conclude and discuss some future research directions.

\section{The Model} \label{sec:mod}

In this section, we formally define the PM model and describe its functionality. In what follows, we denote by $G=(V,E)$ the (directed) communication graph: $V$ is the set of agents, or \emph{population}, and $E$ is the set of permissible ordered pairwise interactions between these agents. We provide definitions for general communication graphs and unbounded memories, although in this work we deal with complete communication graphs only and are mainly interested in computations that are space-bounded by a logarithm of the population size. We generally denote by $n$ the population size (i.e. $n\equiv |V|$).

\begin{definition}
A \emph{PM} protocol is a 6-tuple $(X,\Gamma,Q,\delta,\gamma,q_0)$ where $X$, $\Gamma$ and $Q$ are all finite sets and
\begin{enumerate}
\item $X$ is the \emph{input alphabet}, where $\sqcup\notin X$,
\item $\Gamma$ is the \emph{tape alphabet}, where $\sqcup\in \Gamma$ and $X\subset \Gamma$,
\item $Q$ is the set of \emph{states},
\item $\delta:Q\times\Gamma^4\rightarrow Q\times\Gamma^4\times\{L,R\}^4\times\{0,1\}$ is the \emph{internal transition function},
\item $\gamma:Q\times Q\rightarrow Q\times Q$ is the \emph{external transition function} (or \emph{interaction transition function}), and
\item $q_0\in Q$ is the \emph{initial state}.
\end{enumerate}
\end{definition}

Each agent is equipped with the following:
\begin{itemize}
\item A \emph{sensor} in order to sense its environment and receive a piece of the input.
\item Four read/write \emph{tapes}: the \emph{working tape}, the \emph{output tape}, the \emph{incoming message tape} and the \emph{outgoing message tape}. We assume that all tapes are bounded to the left and unbounded to the right.
\item A \emph{control unit} that contains the state of the agent and applies the transition functions.
\item Four \emph{heads} (one for each tape) that read from and write to the cells of the corresponding tapes and can move one step at a time, either to the left or to the right.
\item A binary \emph{working flag} either set to $1$ meaning that the agent is \emph{working} internally or to $0$ meaning that the agent is \emph{ready} for interaction.
\end{itemize}

Initially, all agents are in state $q_0$, their working flag is set to $1$, and all their cells contain the \emph{blank symbol} $\sqcup$. We assume that all agents concurrently receive their sensed input (different agents may sense different data) as a response to a global start signal. The input to each agent is a symbol from $X$ and is written on the leftmost cell of its working tape.

When its working flag is set to 1 we can think of an agent working as a usual multitape Turing Machine (but it additionally writes the working flag). In particular, while the working flag is set to 1 the internal transition function $\delta$ is applied, the control unit reads the symbols under the heads and its own state, updates all of them, moves each head one step to the left or to the right, and sets the working flag to 0 or 1, according to $\delta$.

As it is common in the population protocol literature, a \emph{fair adversary scheduler} selects ordered pairs of agents (edges from $E$) to interact. Assume now that two agents $u$ and $\upsilon$ are about to interact with $u$ being the \emph{initiator} of the interaction and $\upsilon$ being the \emph{responder}. Let $f:V\rightarrow \{0,1\}$ be a function returning the current value of each agent's working flag. If at least one of $f(u)$ and $f(\upsilon)$ is equal to $1$, then nothing happens, because at least one agent is still working internally. Otherwise ($f(u)=f(\upsilon)=0$), both agents are ready and an \emph{interaction} is established. In the latter case, the external transition function $\gamma$ is applied, the states of the agents are updated accordingly, the outgoing message of the initiator is copied to the leftmost cells of the incoming message tape of the responder (replacing its contents and writting $\sqcup$ to all other previously non-blank cells) and vice versa (we call this the \emph{message swap}), and finally the working flags of both agents are again set to 1. These operations could be handled by the protocols themselves, but then protocol descriptions would become awkward. So, we simply think of them as automatic operations performed by the hardware. These operations are also considered as atomic, that is, the interacting agents cannot take part in another interaction before the completion of these operations and, moreover, either all operations totally succeed or are all totally aborted (in which case, the states of the interacting agents are restored).


Note that the assumption that the internal transition function $\delta$ is only applied when the working flag is set to 1 is weak. In fact, an equivalent way to model this is to assume that $\delta$ is of the form $\delta:Q\times\Gamma^4\times\{0,1\}\rightarrow Q\times\Gamma^4\times\{L,R,S\}^4\times\{0,1\}$, that it is always applied, and that for all $q\in Q$ and $a\in\Gamma^4$, $\delta(q,a,0)=(q,a,S^4,0)$ is satisfied, where $S$ means that the corresponding head ``stays put''. The same holds for the assumptions that $\gamma$ is not applied if at least one of the interacting agents is working internally and that the working flags are set to $1$ when some established interaction comes to an end; it is equivalent to an extended $\gamma$ of the form $\gamma:Q^2\times \{0,1\}^2\rightarrow Q^2\times  \{0,1\}^2$, that is applied in every interaction, and for which $\gamma(q_1,q_2,f_1,f_2)=(q_1,q_2,f_1,f_2)$ if $f_1=1$ or $f_2=1$, and $\gamma(q_1,q_2,f_1,f_2)=(\gamma_1(q_1,q_2),\gamma_2(q_1,q_2),1,1)$ if $f_1=f_2=0$, hold for all $q_1,q_2\in Q$, and we could also further extend $\delta$ and $\gamma$ to handle the exchange of messages, but for sake of simplicity we have decided to leave such details out of the model.

Since each agent is a TM, we use the notion of a configuration to capture its ``state''. An \emph{agent configuration} is a tuple $(q,l_w,r_w,l_o,r_o,l_{im},r_{im},l_{om},r_{om},f)$, where $q\in Q$, $l_i,r_i\in \Gamma^{*}$, and $f\in \{0,1\}$. $q$ is the state of the control unit, $l_w$ ($l_o, l_{im},l_{om}$) is the string of the working (output, incoming message, outgoing message) tape to the left of the head (including the symbol scanned), $r_w$ ($r_o, r_{im}, r_{om}$) is the string of the working (output, incoming message, outgoing message) tape to the right of the head  (excluding infinite sequences of blank cells), and $f$ is the working flag indicating whether the agent is ready to interact ($f=0$) or carrying out some internal computation ($f=1$). Let $\mathcal{B}$ be the set of all agent configurations. Given two agent configurations $A,A^{\prime}\in \mathcal{B}$, we say that $A$ \emph{yields} $A^{\prime}$ if $A^{\prime}$ follows $A$ by a single application of $\delta$.

A \emph{population configuration} is a mapping $C:V\rightarrow \mathcal{B}$, specifying the agent configuration of each agent in the population. Let $C$, $C^{\prime}$ be population configurations and let $u\in V$. We say that $C$ \emph{yields} $C^{\prime}$ via \emph{agent transition} $u$, denoted $C \stackrel{u}\rightarrow C^{\prime}$, if  $C(u)$ yields $C^{\prime}(u)$ and $C^{\prime}(w)=C(w)$, $\forall w\in V-\{u\}$.

Denote by $q(A)$ the state component of an agent configuration $A$ and similarly for the other components (e.g. $l_w(A)$, $r_{im}(A)$, $f(A)$, and so on). Let $s_{tp}(A)=l_{tp}(A)r_{tp}(A)$, that is, we obtain by concatenation the whole contents of tape $tp\in \{w,o,im,om\}$. Given a string $s$ and $1\leq i,j\leq |s|$ denote by $s[\ldots i]$ its prefix $s_1s_2\cdots s_i$ and by $s[j\ldots]$ its suffix $s_js_{j+1}\cdots s_{|s|}$. If $i,j>|s|$ then $s[\ldots i]=s\sqcup^{i-|s|}$ (i.e. $i-|s|$ blank symbols appended to $s$) and $s[j\ldots]=\varepsilon$. For any external transition $\gamma(q_1,q_2)=(q_1^{\prime},q_2^{\prime})$ define $\gamma_1(q_1,q_2)=q_1^{\prime}$ and $\gamma_2(q_1,q_2)=q_2^{\prime}$. Given two population configurations $C$ and $C^{\prime}$, we say that $C$ \emph{yields} $C^{\prime}$ via \emph{encounter} $e=(u,\upsilon)\in E$, denoted $C\stackrel{e}\rightarrow C^{\prime}$, if one of the following two cases holds:\\

\noindent Case 1 (only for this case, we define $C^u\equiv C(u)$ to avoid excessive number of parentheses):
\begin{itemize}
\item $f(C(u))=f(C(\upsilon))=0$, which guarantees that both agents $u$ and $\upsilon$ are ready for interaction under the population configuration $C$.
\item $C^{\prime}(u) = (\gamma_1(q(C^u),q(C^\upsilon)),l_w(C^u), r_w(C^u),l_o(C^u), r_o(C^u), s_{om}(C^\upsilon)[\ldots |l_{im}(C^u)|],$\\ \hspace*{14.5mm} $s_{om}(C^\upsilon)[|l_{im}(C^u)|+1\ldots], l_{om}(C^u), r_{om}(C^u), 1)$,
\item $C^{\prime}(\upsilon) = (\gamma_2(q(C^u),q(C^\upsilon)),l_w(C^\upsilon), r_w(C^\upsilon),l_o(C^\upsilon), r_o(C^\upsilon), s_{om}(C^u)[\ldots |l_{im}(C^\upsilon)|],$\\ \hspace*{14.5mm} $s_{om}(C^u)[|l_{im}(C^\upsilon)|+1\ldots], l_{om}(C^\upsilon), r_{om}(C^\upsilon), 1)$, and
\item $C^{\prime}(w)=C(w)$, $\forall w\in V-\{u,\upsilon\}$.
\end{itemize}

\noindent Case 2:
\begin{itemize}
\item $f(C(u))=1$ or $f(C(\upsilon))=1$, which means that at least one agent between $u$ and $\upsilon$ is working internally under the population configuration $C$, and
\item $C^{\prime}(w) = C(w)$, $\forall w\in V$. In this case no effective interaction takes place, thus the population configuration remains the same.
\end{itemize}

Generally, we say that $C$ \emph{yields} (or \emph{can go in one step to}) $C^{\prime}$, and write $C\rightarrow C^{\prime}$, if $C\stackrel{e}\rightarrow C^{\prime}$ for some $e\in E$ (via encounter) or $C\stackrel{u}\rightarrow C^{\prime}$ for some $u\in V$ (via agent transition), or both. We say that $C^{\prime}$ is \emph{reachable} from $C$, and write $C\stackrel{*}\rightarrow C^{\prime}$ if there is a sequence of population configurations $C=C_0,C_1,\ldots,C_t=C^{\prime}$ such that $C_i\rightarrow C_{i+1}$ holds for all $i\in \{0,1,\ldots,t-1\}$. An \emph{execution} is a finite or infinite sequence of population configurations $C_0,C_1\dots$, so that $C_i\rightarrow C_{i+1}$. An infinite execution is \emph{fair} if for all population configurations $C$, $C^{\prime}$ such that $C\rightarrow C^{\prime}$, if $C$ appears infinitely often then so does $C^{\prime}$. A \emph{computation} is an infinite fair execution.

We assume that the input alphabet $X$, the tape alphabet $\Gamma$, and the set of states $Q$ are all sets whose cardinality is fixed and independent of the population size (i.e. all of them are of cardinality $\mathcal{O}(1)$). Thus, protocol descriptions have also no dependence on the population size and the PM model \emph{preserves uniformity}. Moreover, PM protocols are \emph{anonymous}, since there is no room in the state of the agents for unique identifiers (though here there is plenty of room on the tapes to create such ids). Uniformity and anonymity are two outstanding properties of population protocols \cite{AADFP06}.

\section{Stably Computable Predicates} \label{sec:pred}

The predicates that we consider here are of the following form. The input (also called an \emph{input assignment}) to the population is any string $x=\sigma_1\sigma_2\cdots \sigma_n\in X^*$, with $n$ being the size of the population. In particular, by assuming an ordering over $V$, the input to agent $i$ is the symbol $\sigma_i$, $1\leq i\leq n$. \footnote{For simplicity and w.l.o.g. we have made the assumption that the input to each agent is a single symbol but our definitions can with little effort be modified to allow the agents receive whole strings as input.} Any mapping $p: X^*\rightarrow \{0,1\}$ is a \emph{predicate on input assignments}. 

\begin{definition}
A predicate on input assignments $p$ is called \emph{symmetric} if for every $x\in X^*$ and any $x^{\prime}$ which is a permutation of $x$'s symbols, it holds that $p(x)=p(x^{\prime})$
\end{definition}

In words, permuting the input symbols does not affect the symmetric predicate's outcome. From each predicate $p$ a language $L_p$ is derived that is the set of all strings that make $p$ true or equivalently, $L_p=\{x\in X^*\;|\; p(x)=1\}$. If a predicate $p$ is symmetric, $L_p$ is a symmetric language, that is, \emph{for each input string $x \in L_p$ any permutation of $x$'s symbols $x^{\prime}$ also belongs in $L_p$}.

A population configuration $C$ is called \emph{output stable} if for every configuration $C^{\prime}$ that is reachable from $C$ it holds that $O(C^{\prime}(u))=O(C(u))$ for all $u\in V$, where $O(C(u))\equiv s_o(C(u))$. In simple words, no agent changes its output in any subsequent step and no matter how the computation proceeds. A predicate on input assignments $p$ is said to be \emph{stably computable} by a PM protocol $\mathcal{A}$ in a graph family $\mathcal{U}$ if, for any input assignment $x\in X^*$, any computation of $\mathcal{A}$, on any communication graph from $\mathcal{U}$ of order $|x|$, contains an output stable configuration in which all agents have $p(x)$ written on their output tape. In what follows, we always assume that the graph family under consideration contains only complete communication graphs.

We say that a predicate $p$ over $X^*$ belongs to $SPACE(f(n))$ ($NSPACE(f(n))$) if there exists some deterministic (nondeterministic, resp.) TM that decides $L_p$ using $\mathcal{O}(f(n))$ space.

\begin{definition}
Let $SSPACE(f(n))$ and $SNSPACE(f(n))$ be $SPACE(f(n))$'s and $NSPACE(f(n))$'s restrictions to symmetric predicates, respectively.
\end{definition}

\begin{definition}
Let $PMSPACE(f(n))$ be the class of all predicates that are stably computable by some PM protocol that uses $\mathcal{O}(f(n))$ space in every agent (and in all of its tapes).
\end{definition}

\begin{remark} \label{rem:sym}
All agents are identical and do not initially have unique ids, thus, stably computable predicates by the PM model on complete communication graphs have to be symmetric.
\end{remark}

In fact, although we provide general definitions, we are mainly interested in the class $PLM \equiv PMSPA\-CE(\log n)$ and we call a PM protocol a \emph{PALOMA} protocol (standing for PAssively mobile LOgarithmic space MAchines) if it always uses $\mathcal{O}(\log n)$ space. Our main result in this paper is the following exact characterization for $PLM$: $PLM$ is equal to $SNSPACE(n\log n)$. In fact, the proof we give easily generalizes and also provides an exact characterization for $PMSPACE(f(n))$, when $f(n)=\Omega(\log n)$: $PMSPACE(f(n))$ is equal to $SNSPACE(nf(n))$.

\section{Two Examples \& a First Inclusion for $PLM$} \label{exam}

\subsection{Multiplication of Variables} \label{app:mult1}

We present now a PM protocol that stably computes the predicate $(N_c=N_a\cdot N_{b})$ using $\mathcal{O}(\log n)$ space (on the complete communication graph of $n$ nodes) that is, all agents eventually decide whether the number of $c$s in the input assignment is the product of the number of $a$s and the number of $b$s. We give a high-level description of the protocol.

Initially, all agents have one of $a$, $b$ and $c$ written on the first cell of their working memory (according to their sensed value). That is, the set of input symbols is $X=\Sigma=\{a,b,c\}$. Each agent that receives input $a$ goes to state $a$ and becomes ready for interaction (sets its working flag to 0). Agents in state $a$ and $b$ both do nothing when interacting with agents in state $a$ and agents in state $b$. An agent in $c$ initially creates in its working memory three binary counters, the $a$-counter that counts the number of $a$s, the $b$-counter, and the $c$-counter, initializes the $a$ and $b$ counters to 0, the $c$-counter to 1, and becomes ready. When an agent in state $a$ interacts with an agent in state $c$, $a$ becomes $\bar{a}$ to indicate that the agent is now sleeping, and $c$ does the following (in fact, we assume that $c$ goes to a special state $c_{a}$ in which it knows that it has seen an $a$, and that all the following are done internally, after the interaction; finally the agent restores its state to $c$ and becomes again ready for interaction): it increases its $a$-counter by one (in binary), multiplies its $a$ and $b$ counters, which can be done in binary in logarithmic space (binary multiplication is in $LOGSPACE$), compares the result with the $c$-counter, copies the result of the comparison to its output tape, that is, 1 if they are equal and 0 otherwise, and finally it copies the comparison result and its three counters to the outgoing message tape and becomes ready for interaction. Similar things happen when a $b$ meets a $c$ (interchange the roles of $a$ and $b$ in the above discussion). When a $c$ meets a $c$, the responder becomes $\bar{c}$ and copies to its output tape the output bit contained in the initiator's message. The initiator remains to $c$, adds the $a$-counter contained in the responder's message to its $a$-counter, the $b$ and $c$ counters of the message to its $b$ and $c$ counters, respectively, multiplies again the updated $a$ and $b$ counters, compares the result to its updated $c$ counter, stores the comparison result to its output and outgoing message tapes, copies its counters to its outgoing message tape and becomes ready again. When a $\bar{a}$, $\bar{b}$ or $\bar{c}$ meets a $c$ they only copy to their output tape the output bit contained in $c$'s message and become ready again (eg $\bar{a}$ remains $\bar{a}$), while $c$ does nothing.

Note that the number of $c$s is at most $n$ which means that the $c$-counter will become at most $\lceil \log n\rceil$ bits long, and the same holds for the $a$ and $b$ counters, so $\mathcal{O}(\log n)$ memory is required in each tape.

\begin{theorem}
The above PM protocol stably computes the predicate $(N_c=N_a\cdot N_{b})$ using $\mathcal{O}(\log n)$ space.
\end{theorem}
\begin{proof}
Given a fair execution, eventually only one agent in state $c$ will remain, its $a$-counter will contain the total number of $a$s, its $b$-counter the total number of $b$s, and its $c$-counter the total number of $c$s. By executing the multiplication of the $a$ and $b$ counters and comparing the result to its $c$-counter it will correctly determine whether $(N_c=N_a\cdot N_{b})$ holds and it will store the correct result (0 or 1) to its output and outgoing message tapes. At that point all other agents will be in one of the states $\bar{a}$, $\bar{b}$, and $\bar{c}$. All these, again due to fairness, will eventually meet the unique agent in state $c$ and copy its correct output bit (which they will find in the message they get from $c$) to their output tapes. Thus, eventually all agents will output the correct value of the predicate, having used $\mathcal{O}(\log n)$ memory.
\qed
\end{proof}

\begin{corollary}
The class of semilinear predicates is a proper subset of $PLM$.
\end{corollary}
\begin{proof}
PALOMA protocols simulate population protocols and $(N_c=N_a\cdot N_{b})\in PLM$, which is non-semilinear.
\qed
\end{proof}

\subsection{Power of 2} \label{app:pow}

Here, we present a PM protocol that, using $\mathcal{O}(\log n)$ memory, stably computes the predicate $(N_1=2^{t})$, where $t\in \mathbb{Z}_{\geq 0}$, on the complete communication graph of $n$ nodes, that is, all agents eventually decide whether the number of $1$s in the input assignment is a power of 2.

The set of input symbols is $\Sigma=\{0,1\}$. All agents that receive 1 create a binary $1$-counter to their working tape and initialize it to $1$. Moreover, they create a binary $next\_pow\_of2$ block and set it to $2$. Finally, they write 1 (which is interpreted as ``true'') to their output tape, and copy the $1$-counter and the output bit to their outgoing message tape before going to state 1 and becoming ready. Agents that receive input 0 write 0 (which is interpreted as ``false'') to their output tape, go to state 0, and become ready. Agents in state 0 do nothing when interacting with each other. When an agent in state 0 interacts with an agent in state 1, then 0 simply copies the output bit of 1. When two agents in state 1 interact, then the following happen: the initiator sets its $1$-counter to the sum of the responder's $1$-counter and its own $1$-counter and compares its updated value to $next\_pow\_of2$. If it is less than $next\_pow\_of2$ then it writes 0 to the output tape. If it is equal to $next\_pow\_of2$ it sets $next\_pow\_of2$ to $2\cdot next\_pow\_of2$ and sets its output bit (in the output tape) to 1. If it is greater than $next\_pow\_of2$, then it starts doubling $next\_pow\_of2$ until $1$-counter $\geq next\_pow\_of2$ is satisfied. If it is satisfied by equality, then it doubles $next\_pow\_of2$ one more time and writes 1 to the output tape. Otherwise, it simply writes 0 to the output tape. Another implementation would be to additionally send $next\_pow\_of2$ blocks via messages and make the initiator set $next\_pow\_of2$ to the maximum of its own and the responder's $next\_pow\_of2$ blocks. In this case at most one doubling would be required. Finally, in both implementations, the initiator copies the output bit and the $1$-counter to its outgoing message tape (in the second implementation it would also copy $next\_pow\_of2$ to the outgoing message tape), remains in state $1$, and becomes ready. The responder simply goes to state $\bar{1}$ and becomes ready. An agent in state $\bar{1}$ does nothing when interacting with an agent in state 0 and vice versa. When an agent in state $\bar{1}$ interacts with an agent in state 1, then $\bar{1}$ simply copies the output bit of 1.

Note that $next\_pow\_of2$ can become at most 2 times the number of $1$s in the input assignment, and the latter is at most $n$. Thus, it requires at most $\lceil \log 2n\rceil$ bits of memory. Either way, we can delay the multiplication until another 1 appears, in which case we need at most $\lceil \log n\rceil$ bits of memory for storing $next\_pow\_of2$ (the last unnecessary multiplication will never be done).

\section{Assigning Unique IDs by Reinitiating Computation} \label{sec:uids}

In this section, we first prove that PM protocols can assume the existence of unique consecutive ids and knowledge of the population size at the space cost of $\mathcal{O}(\log n)$ (Theorem \ref{the:iplm}). In particular, we present a PM protocol that correctly assigns unique consecutive ids to the agents and informs them of the correct population size using only $\mathcal{O}(\log n)$ memory, without assuming any initial knowledge of none of them. We show that this protocol can simulate any PM protocol that assumes the existence of these ids and knows the population size. At the end of the section, we exploit this result to prove that $SSPACE(n\log n) \subseteq PLM$.

Pick any $p\in SIPLM$. Let $\mathcal{A}$ be the IPM protocol that stably computes it in $\mathcal{O}(\log n)$ space. We now present a PM protocol $\mathcal{I}$, containing protocol $\mathcal{A}$ as a subroutine (see Protocol \ref{prot:ids}), that stably computes $p$, by also using $\mathcal{O}(\log n)$ space. $\mathcal{I}$ is always executed and its job is to assign unique ids to the agents, to inform them of the correct population size and to control $\mathcal{A}$'s execution (e.g. restarts its execution if needed). $\mathcal{A}$, when $\mathcal{I}$ allows its execution, simply reads the unique ids and the population size provided by $\mathcal{I}$ and executes itself normally. We first present $\mathcal{I}$ and then prove that it eventually correctly assigns unique ids and correctly informs the agents of the population size, and that when this process comes to a successful end, it restarts $\mathcal{A}$'s execution in all agents without allowing non-reinitialized agents to communicate with the reinitialized ones. Thus, at some point, $\mathcal{A}$ will begin its execution reading the correct unique ids and the correct population size (provided by $\mathcal{I}$), thus, it will get correctly executed and will stably compute $p$.

We begin by describing $\mathcal{I}$'s variables. $id$ is the variable storing the id of the agent (from which $\mathcal{A}$ reads the agents' ids), $sid$ the variable storing the $id$ that an agent writes in its outgoing message tape in order to send it, and $rid$ the variable storing the $id$ that an agent receives via interaction. Recall the model's convention that all variables used for sending information, like $sid$, preserve their value in future interactions unless altered by the agent. Initially, $id=sid=0$ for all agents. All agents have an input backup variable $binput$ which they initially set to their input symbol and make it read-only. Thus, each agent has always available its input via $binput$ even if the computation has proceeded. $working$ represents the block of the working tape that $\mathcal{A}$ uses for its computation and $output$ represents the contents of the output tape. $initiator$ is a binary flag that after every interaction becomes true if the agent was the initiator of the interaction and false otherwise (this is easily implemented by exploiting the external transition function). $ps$ is the variable storing the population size, $sps$ the one used to put it in a outgoing message, and $rps$ the received one. Initially, $ps=sps=1$.

We now describe $\mathcal{I}$'s functionality. Whenever a pair of agents with the same id interact, the initiator increases its id by one and both update their population size value to the greater id plus one. Whenever two agents with different ids and population size values interact, they update their population size variables to the greater size. Thus the correct size (greatest id plus one) is propagated to all agents. Both interactions described above reinitialize the participating agents (restore their input and erase all data produced by the subroutine $\mathcal{A}$, without altering their ids and population sizes). $\mathcal{A}$, runs as a subroutine whenever two agents of different ids and same population sizes interact, using those data provided by $\mathcal{I}$.

\algsetup{indent=2em}
\floatname{algorithm}{Protocol}
\renewcommand{\algorithmiccomment}[1]{// #1}
\begin{algorithm}[H]
  \caption{$\mathcal{I}$}\label{prot:ids}
  \begin{algorithmic}[1]
    \medskip
    \IF [two agents with the same ids interact]{$rid==id$}
       \IF [the initiator]{$initiator==1$}
          \STATE $id\leftarrow id+1$, $sid\leftarrow id$ \COMMENT {increases its id by one and writes it in the outgoing message tape}
	  \STATE $ps\leftarrow id+1$, $sps\leftarrow ps$ \COMMENT {sets the population size equal to its updated id plus 1}
       \ELSE [the responder]
	  \STATE $ps\leftarrow id+2$, $sps\leftarrow ps$
       \ENDIF
       \STATE \COMMENT {both clear their working block and copy their input symbol into it}
       \STATE \COMMENT {they also clear their output tape}
       \STATE $working\leftarrow binput$, $output\leftarrow \emptyset$
    \ELSE [two agents whose ids differ interact]
      \IF [the one who knows an outdated population size]{$rps > ps$}
         \STATE $working\leftarrow binput$, $output\leftarrow \emptyset$ \COMMENT {is reinitialized}
         \STATE $ps\leftarrow rps$, $sps\leftarrow ps$ \COMMENT {and updates its population size to the greater value}
	  \ELSIF [they know the same population size] {$rps == ps$}
	  	\STATE \COMMENT {so they are both reinitialized and can proceed executing $\mathcal{A}$}
	    \PRINT $\mathcal{A}$ for $1$ step
	  \ENDIF
    \ENDIF
  \end{algorithmic}
\end{algorithm}

\begin{lemma} \label{lem:ids1}
(i) No agent $id$ will get greater than $n-1$, and no $ps$ variable will get greater than $n$. (ii) $\mathcal{I}$ assigns the ids $\{0,1,\ldots,n-1\}$ in a finite number of interactions. (iii) $\mathcal{I}$ sets the ps variable of each agent to the correct population size in a finite number of interactions.
\end{lemma}

\begin{proof}
(i) By an easy induction, in order for an $id$ to reach the value $v$, there have to be at least $v+1$ agents present in the population. Thus, whenever an $id$ gets greater than $n-1$, there have to be more than $n$ agents present, which creates a contradiction. Similar arguments hold for the $ps$ variables\\
(ii) Assume on the contrary that it does not. Because of (i), at each point of the computation there will exist at least two agents, $u$, $v$ such that $id_u=id_v$. Due to fareness, an interaction between such agents shall take place infinitely many times, creating an arbitrarily large $id$ which contradicts (i).\\
(iii) The correctness of the id assignment ((i),(ii)) guarantees that after a finite number of steps two agents, $u$, $v$ will set their $ps$ variables to the correct population size (upon interaction in which $id_u=id_v=n-2$). It follows from (i) that no agent will have its $ps$ variable greater than $n$. Fareness guarantees that each other agent will interact with $u$ or $v$, updating its $ps$ to $n$.
\qed
\end{proof}

\begin{lemma} \label{lem:ids5}
Given that $\mathcal{I}$'s execution is fair, $\mathcal{A}$'s execution is fair as well.
\end{lemma}
\begin{proof}

Due to the fact that the id-assignment process and the population size propagation are completed in a finite number of steps, it suffices to study fairness of $\mathcal{A}$'s execution after their completion. The state of each agent may be thought of as containing an $\mathcal{I}$-subcomponent and an $\mathcal{A}$-subcomponent, with obvious contents. Denote by $C_{\mathcal{A}}$ the unique subconfiguration of $C$ consisting only of the $\mathcal{A}$-subcomponents of all agents and note that some $C_{\mathcal{A}}$ may correspond to many superconfigurations $C$. Assume that $C_{\mathcal{A}}\rightarrow C^{\prime}_{\mathcal{A}}$ and that $C_{\mathcal{A}}$ appears infinitely often (since here we consider $\mathcal{A}$'s configurations, this `$\rightarrow$' refers to a step of $\mathcal{A}$'s execution). $C_{\mathcal{A}}\rightarrow C^{\prime}_{\mathcal{A}}$ implies that there exist superconfigurations $C$, $C^{\prime}$ of $C_{\mathcal{A}}$, $C^{\prime}_{\mathcal{A}}$, respectively, such that $C\rightarrow C^{\prime}$ (via some step of $\mathcal{A}$ in the case that $C_{\mathcal{A}}\neq C^{\prime}_{\mathcal{A}}$). Due to $\mathcal{I}$'s fairness, if $C$ appears infinitely often, then so does $C^{\prime}$ and so does $C^{\prime}_{\mathcal{A}}$ since it is a subconfiguration of $C^{\prime}$. Thus, it remains to show that $C$ appears infinitely often. Since $C_{\mathcal{A}}$ appears infinitely often, then the same must hold for all of its superconfigurations. The reasoning is as follows. All those superconfigurations differ only in the $\mathcal{I}$-subcomponents, that is, they only differ in some variable checks performed by $\mathcal{I}$ (after the id-assignment process and the population size propagation have come to an end, nothing else is performed by $\mathcal{I}$). But all of them are reachable from and can reach a common superconfiguration of $C_{\mathcal{A}}$ in which no variable checking is performed by $\mathcal{I}$, thus, they only depend on which pair of agents is selected for interaction and they are all reachable from one another. Since at least one of them appears infinitely often then, due to the fairness of $\mathcal{I}$'s execution, all of them must also appear infinitely often and this completes the proof.
\qed
\end{proof}

By combining the above lemmata we can prove the following:

\begin{theorem} \label{the:iplm}
$PLM=SIPLM$.
\end{theorem}
\begin{proof}
$PLM\subseteq SIPLM$ holds trivially, so it suffices to show that $SIPLM\subseteq PLM$. We have presented a $PLM$ protocol (protocol \ref{prot:ids}) that assigns the agents unique consecutive ids after a finite number of interactions and informs them of the population size (lemma \ref{lem:ids1}). It follows directly from the protocol that after that point, further fair execution of $\mathcal{I}$ will result in execution of protocol $\mathcal{A}$ which can take into account the existance of unique ids. Moreover, execution of $\mathcal{A}$ is guaranteed to be fair (lemma \ref{lem:ids5}).
\qed
\end{proof}

\section{Exploring $PLM$}\label{plm}

\subsection{An upper bound}
We first prove that $PLM \subseteq NSPACE(n\log n)$.

\begin{theorem}\label{the:plmup}
 All predicates in $PLM$ are in the class $NSPACE(n\log n)$.
\end{theorem}
 
\begin{proof}
Let $\mathcal{A}$ be a PM protocol that stably computes such a predicate $p$ using $\mathcal{O}(\log n)$ memory. A population configuration can be represented as an $n-$place vector storing an agent configuration per place, and thus uses $\mathcal{O}(n\log n)$ space in total. The language $L_p$ derived from $p$ is the set of all strings that make $p$ true, that is, $L_p=\{x\in X^*\;|\; p(x)=1\}$.

We will now present a nondeterministic Turing Machine $\mathcal{M_A}$ that decides $L_p$ in $\mathcal{O}(n\log n)$ space. To accept the input (assignment) $x$, $\mathcal{M_A}$ must verify two conditions: That there exists a configuration $C$ reachable from the initial configuration corresponding to $x$ in which the output tape of each agent indicates that $p$ holds, and that there is no configuration $C^{\prime}$ reachable from $C$ under which $p$ is violated for some agent.

The first condition is verified by guessing and checking a sequence of configurations. Starting from the initial configuration, each time $\mathcal{M_A}$ guesses configuration $C_{i+1}$ and verifies that $C_i$ yields $C_{i+1}$. This can be caused either by an agent transition $u$, or an encounter $(u,v)$. In the first case, the verification can be carried out as follows: $\mathcal{M_A}$ guesses an agent $u$ so that $C_i$ and $C_{i+1}$ differ in the configuration of $u$, and that $C_i(u)$ yields $C_{i+1}(u)$.  It then verifies that $C_i$ and $C_{i+1}$ differ in no other agent configurations. Similarly, in the second case $\mathcal{M_A}$ nondeterministically chooses agents $u$, $v$ and verifies that encounter $(u,v)$ leads to $C^{\prime}$ by ensuring that: (a) both agents have their working flags cleared in $C$, (b) the tape exchange takes place in $C^{\prime}$, (c) both agents update their states according to $\gamma$ and set their working flags to $1$ in $C^{\prime}$ and (d) that $C_i$ and $C_{i+1}$ differ in no other agent configurations. In each case, the space needed is $\mathcal{O}(n\log n)$ for storing $C_i$, $C_{i+1}$, plus $\mathcal{O}(\log n)$ extra capacity for ensuring the validity of each agent configuration in $C_{i+1}$.

If the above hold, $\mathcal{M_A}$ replaces $C_i$ with $C_{i+1}$ and repeats this step. Otherwise, $\mathcal{M_A}$ drops $C_{i+1}$. Any time a configuration $C$ is reached in which $p$ holds, $\mathcal{M_A}$ computes the complement of a similar reachability problem: it verifies that there exists no configuration reachable from $C$ in which $p$ is violated. Since $NSPACE$ is closed under complement for all space functions $\geq \log n$ (see Immerman-Szelepcs\'enyi theorem, \cite{Pa94}, pages $151-153$), this condition can also be verified in $\mathcal{O}(n\log n)$ space. Thus, $L_p$ can be decided in $\mathcal{O}(n\log n)$ space by some nondeterministic Turing Machine, so $L_p\in NSPACE(n\log n)$.
\qed
\end{proof}

\subsection{A lower bound}
We now prove that $SNSPACE(n\log n)$ is a subset of $PLM$. We establish this result by first proving a weaker one, that is, that $SSPACE(n\log n)$ is a subset of $PLM$.

\begin{theorem} \label{the:psize}
$SSPACE(n\log n) \subseteq PLM$.
\end{theorem}
\begin{proof}
Let $p:X^*\rightarrow \{0,1\}$ be any predicate in $SSPACE(n\log n)$ and $\mathcal{M}$ be the deterministic TM that decides $p$ by using $\mathcal{O}(n\log n)$ space. We construct a PM protocol $\mathcal{A}$ that stably computes $p$ by exploiting its knowledge of the population size. Let $x$ be any input assignment in $X^*$. Each agent receives its input symbol according to $x$ (e.g. $u$ receives symbol $x(u)$). Now the agents obtain unique ids according to protocol \ref{prot:ids}. The agent that has obtained the unique id $0$ starts simulating $\mathcal{M}$.

Assume that currently the simulation is carried out by an agent $u$ having the id $i_u$. Agent $u$ uses its simulation tape to write symbols according to the transition function of $\mathcal{M}$. Any time the head of $\mathcal{M}$ moves to the left, $u$ moves the head of the simulation tape to the left, pauses the simulation, writes the current state of $\mathcal{M}$ to its outgoing message tape, and passes the simulation to the agent $v$ having id $i_v= (i_u-1) \mod n$. Similarly, any time the head of $\mathcal{M}$ moves to the right, $u$ moves the head of the simulation tape to the right, pauses the simulation, writes the current state of $\mathcal{M}$ to its outgoing message tape, and passes the simulation to the agent $v$ having id $i_v= (i_u+1) \mod n$. In both cases, agent $v$ copies the state of $\mathcal{M}$ to its working tape, and starts the simulation.

Whenever, during the simulation, $\mathcal{M}$ accepts, then $\mathcal{A}$ also accepts; that is, the agent that detects $\mathcal{M}$'s acceptance, writes 1 to its output tape and informs all agents to accept. If $\mathcal{M}$ rejects, it also rejects. Finally, note that $\mathcal{A}$ simulates $\mathcal{M}$ not necessarily on input $x=(s_0,s_1,\ldots,s_{n-1})$ but on some $x^{\prime}$ which is a permutation of $x$. The reason is that agent with id $i$ does not necessarily obtain $s_{i}$ as its input. The crucial remark that completes the proof is that $\mathcal{M}$ accepts $x$ if and only if it accepts $x^{\prime}$, because $p$ is symmetric.

Because of the above process, it is easy to verify that the $k-$th cell of the simulation tape of any agent $u$ having the id $i_u$ corresponds to the $n\cdot k + i_u-$th cell of $\mathcal{M}$. Thus, whenever $\mathcal{M}$ alters $l=\mathcal{O}(n\cdot \log n)$ tape cells, any agent $u$ will alter $l^{\prime}= \frac{l-i_u}{n} =\mathcal{O}(\log n)$ cells of its simulation tape.
\qed
\end{proof}

We now show how the above approach can be generalized to include nondeterministic Turing Machines.
\begin{theorem} \label{the:lowPLM}
$SNSPACE(n\log n) \subseteq PLM$.
\end{theorem}
\begin{proof}
By considering Theorem \ref{the:iplm}, it suffices to show that $SNSPACE(n\log n)$ is a subset of $SIPLM$. We have already shown that the IPM model can simulate a deterministic TM $\mathcal{M}$ of $\mathcal{O}(n\log n)$ space by using $\mathcal{O}(\log n)$ space (Theorem \ref{the:psize}). We now present some modifications that will allow us to simulate a nondeterministic TM $\mathcal{N}$ of the same memory size. Keep in mind that $\mathcal{N}$ is a decider for some predicate in $SNSPACE(n\log n)$, thus, it always halts. Upon initialization, each agent enters a reject state (writes $0$ to its output tape) and the simulation is carried out as in the case of $\mathcal{M}$.

Whenever a nondeterministic choice has to be made, the corresponding agent gets ready and waits for participating in an interaction. The id of the other participant will provide the nondeterministic choice to be made. One possible implementation of this idea is the following. Since there is a fixed upper bound on the number of nondeterministic choices (independent of the population size), the agents can store them in their memories. Any time a nondeterministic choice has to be made between $k$ candidates the agent assigns the numbers $0,1,\ldots,k-1$ to those candidates and becomes ready for interaction. Assume that the next interaction is with an agent whose id is $i$. Then the nondeterministic choice selected by the agent is the one that has been assigned the number $i \mod k$. Fairness guarantees that, in this manner, all possible paths in the tree representing $\mathcal{N}$'s nondeterministic computation will eventually be followed.

Any time the simulation reaches an accept state, all agents change their output to 1 and the simulation halts. Moreover, any time the simulation reaches a reject state, it is being reinitiated. The correctness of the above procedure is captured by the following two cases.
\begin{enumerate}
\item \emph{If $\mathcal{N}$ rejects then every agent's output stabilizes to $0$}. Upon initialization, each agent's output is $0$ and can only change if $\mathcal{N}$ reaches an accept state. But all branches of $\mathcal{N}$'s computation reject, thus, no accept state is ever reached, and every agent's output forever remains to $0$.
\item \emph{If $\mathcal{N}$ accepts then every agent's output stabilizes to $1$}. Since $\mathcal{N}$ accepts, there is a sequence of configurations $S$, starting from the initial configuration $C$ that leads to a configuration $C^{\prime}$ in which each agent's output is set to $1$ (by simulating directly the branch of $\mathcal{N}$ that accepts). Notice that when an agent sets its output to $1$ it never alters its output tape again, so it suffices to show that the simulation will eventually reach $C^{\prime}$. Assume on the contrary that it does not. Since $\mathcal{N}$ always halts the simulation will be at the initial configuration $C$ infinitely many times. Due to fairness, by an easy induction on the configurations of $S$, $C^{\prime}$ will also appear infinitely many times, which leads to a contradiction. Thus the simulation will eventually reach $C^{\prime}$ and the output will stabilize to $1$.
\end{enumerate}
\qed
\end{proof}

Taking into consideration the above theorems we can conclude to the following exact characterization for the class $PLM$:
\begin{theorem}\label{plm_exact}
$PLM$ is equal to $SNSPACE(n\log n)$.
\end{theorem}
\begin{proof}
Follows from Theorems \ref{the:psize}, which establishes that $SNSPACE(n\log n)\subseteq PLM$, and \ref{the:plmup}, which establishes that $PLM\subseteq NSPACE(n\log n)$; but for all $p\in PLM$, $p$ is symmetric (Remark \ref{rem:sym}), thus, $PLM\subseteq SNSPACE(n\log n)$.
\qed
\end{proof}

\subsection{Simulating Nondeterministic Recognizers}

Here, we generalize the preceding ideas to nondeterministic \emph{recognizers} of $\mathcal{O}(n\log n)$ space. There is a way to stably compute predicates in $SSPACE(n\log n)$ even when the corresponding TM $N$ might loop, by carrying out an approach similar to the one given above. However, since neither an accept nor a reject state may be reached, the simulation is nondeterministically reinitiated at any point that is not in such a state. This choice is also obtained by the nondeterministic interactions. For example, whenever the agent that carries out the simulation interacts with an agent that has an id that is even, the simulation remains unchanged, otherwise it is reinitiated. Notice however that during the simulation, any agent having id $i$ may need to interact with those having neighboring ids, so those must not be able to cause a reinitiation in the simulation.

Correctness of the above procedure is captured by similar arguments to those in the proof of Theorem \ref{the:lowPLM}. If $\mathcal{N}$ never accepts, then no output tape will ever contain a $1$, so the simulation stabilizes to $0$. If $\mathcal{N}$ accepts there is a sequence of configurations $S$, starting from the initial configuration $C$ that leads to a configuration $C^{\prime}$ in which each agent's output is set to $1$. Observe that this is a ``good'' sequence, meaning that no reinitiations take place, and, due to fairness, it will eventually occur.

\section{Space Hierarchy of the PM Model}
In this section we study the behaviour of PM model for various space bounds.


\begin{theorem} \label{the:plmfn}
For any function $f:\mathbb{N} \rightarrow \mathbb{N}$, any predicate in $PMSPACE(f(n))$ is also in $SNSPACE(2^{f(n)}($ $f(n) + \log n))$.
\end{theorem}
\begin{proof}
Take any $p\in PMSPACE(f(n))$. Let $\mathcal{A}$ be the PM protocol that stably computes predicate $p$ in space $\mathcal{O}(f(n))$. $L_p=\{(s_1,s_2,\ldots,s_n)\;|\; s_i\in X \mbox{ for all } i\in\{1,\ldots,n\} \mbox{ and }$ $p(s_1,s_2,\ldots,s_n)=1\}$ is the language corresponding to $p$ ($X\subset \Sigma^{*}$ is the set of input strings). We describe a NTM $\mathcal{N}$ that decides $L_p$ in $g(n) = \mathcal{O}(2^{f(n)}(f(n) + \log n))$ space.

Note that each agent uses memory of size $\mathcal{O}(f(n))$. So, by assuming a binary tape alphabet $\Gamma=\{0,1\}$ (the alphabet of the agents' tapes), an assumption which is w.l.o.g., there are $2^{\mathcal{O}(f(n))}$ different agent configurations (internal configurations) each of size $\mathcal{O}(f(n))$. $\mathcal{N}$ stores a population configuration by storing all these agent configurations, consuming for this purpose $\mathcal{O}(f(n)2^{f(n)})$ space, together with a number per agent configuration representing the number of agents in that agent configuration under the current population configuration. These numbers sum up to $n$ and each one of them requires $\mathcal{O}(\log n)$ bits, thus, $\mathcal{O}(2^{f(n)}\log n)$ extra space is needed, giving a total of $\mathcal{O}(2^{f(n)}(f(n) + \log n))$ space needed to store a population configuration. The reason that such a representation of population configurations suffices is that when $k$ agents are in the same internal configuration there is no reason to store it $k$ times. The completeness of the communication graph allows us to store it once and simply indicate the number of agents that are in this common internal configuration, that is, $k$.

Now $\mathcal{N}$ does what the deterministic TM of the Theorem \ref{the:plmup} does. The main difference is that it now store the population configurations according to the new representation we discussed above.
\qed
\end{proof}

\begin{theorem} [Symmetric Space Hierarchy Theorem] \label{the:nsymsh}
For any function $f:\mathbb{N} \rightarrow \mathbb{N}$, a symmetric language $L$ exists that is decidable in $\mathcal{O}(f(n))$ (non)deterministic space but not in $o(f(n))$ (non) deterministic space.
\end{theorem}
\begin{proof}
Follows immediately from the unary (tally) separation language presented in \cite{Ge03} and the fact that any unary language is symmetric.
\qed
\end{proof}

\begin{corollary} \label{cor:ologn}
For any function $f(n)=o(\log n)$ it holds that $PMSPACE(f(n))\subsetneq SN\-SPACE(nf(n))$.
\end{corollary}
\begin{proof}
By considering Theorem \ref{the:plmfn}, the symmetric space hierarchy theorem (Theorem \ref{the:nsymsh}) it suffices to show that $2^{f(n)}(f(n) + \log n)=o(nf(n))$ for $f(n)=o(\log n)$. We have that
\begin{equation*}
2^{f(n)}(f(n) + \log n)=2^{o(\log n)}\mathcal{O}(\log n)=o(n)\mathcal{O}(\log n),
\end{equation*}
which obviously grows slower than $nf(n)=n\cdot o(\log n)$.
\qed
\end{proof}

So, for example, if $f(n)=\log\log n$, then $PMSPACE(\log\log n)\subseteq SNSPA\-CE(\log^{2} n)$ which is strictly smaller than $SNSPACE(nf(n))=SNSPA\-CE(n\log\log n)$ by the symmetric space hierarchy theorem. Another interesting example is obtained by setting $f(n)=c$. In this case, we obtain the $SNSPA\-CE(\log n)=SNL$ (for \emph{Symmetric} $NL$) upper bound for population protocols of Angluin \emph{et al.} \cite{AADFP06} (in \cite{AAE06} they obtained a better bound which is the semilinear predicates and constitutes an exact characterization for population protocols).

Although the above upper bounds are relatively tight for $f(n) \leq \log n$ space functions, the bounds for $f(n) > \log n$ have worse behavior in relation to the $nf(n)$. In the theorem below, tighter upper bounds are shown by working with a different representation of population configurations that holds each agent's internal configuration explicitly.

\begin{theorem} \label{the:ppm_upper}
For any function $f(n)$ it holds that $PMSPACE(f(n))$ is a subset of $SNSPACE(nf(n))$.
\end{theorem}
\begin{proof}
This proof is similar to the proof establishing that $PLM$ is a subset of $SN-SPACE(n\log n)$ in Theorem \ref{the:plmup}. It can be easily shown, using similar arguments to those of the previously stated proof, that all predicates in $PMSPACE(f(n))$ are also in $SNSPACE(nf(n))$. In particular, there is a nondeterministic $TM$ $\mathcal{M}$ of space $\mathcal{O}(nf(n))$ that can decide a language corresponding to any predicate $p \in$ $PMSPACE(f(n))$. $\mathcal{M}$ holds the internal configurations of the $n$ agents each of which needs $\mathcal{O}(f(n))$ space and therefore each configuration (of the entire population) requires $\mathcal{O}(nf(n))$ space in $\mathcal{M}$'s tape. $\mathcal{M}$ starts with the initial configuration $C$, guesses the next $C^{\prime}$ and checks whether it has reached a configuration in which $p$ holds. When $\mathcal{M}$ reaches such a configuration $C$ it computes the complement of a similar reachability problem: it verifies that there exists no configuration reachable from $C$ in which $p$ is violated. This condition can also be verified in $\mathcal{O}(nf(n))$ space since $NSPACE$ is closed under complement for all space functions $g(n) = \Omega(\log n)$ (Immerman-Szelepcs\'enyi theorem \cite{Pa94}, pages 151-153). Note that for any reasonable function $f(n)$, $g(n)=nf(n) \geq \log n$, as required by the Immerman-Szelepcs\'enyi theorem.
\qed
\end{proof}

The upper bounds shown above are obviously better for functions $f(n) = \Omega(\log n)$ than those presented by Theorem \ref{the:plmfn}. Note however, that Theorem \ref{the:ppm_upper} also holds for $f(n) = o(\log n)$ and for those space constructible functions the upper bounds are worse than those of Theorem \ref{the:plmfn}. In order to realize that consider the function $f(n) = c$ (the memory of each agent is independent of the population size, thus this is the basic PP model). According to Theorem \ref{the:ppm_upper} the upper bound is the trivial $SNSPACE(n)$, whereas the Theorem \ref{the:plmfn} decreases the upper bound to $SNSPACE(\log n)$. This behavior is expected due to the configuration representation of the population used by those theorems. When the configuration is stored as $n$-vector where each element of the vector holds the internal configuration of an agent (representation used in Theorem \ref{the:ppm_upper}) then as the memory size grows the additional space needed is a factor $n$ of that growth. On the other hand, when a configuration is represented as a vector of size equal to the number of all possible internal configurations where each element is the number of agents that are in the corresponding internal configuration (as in Theorem \ref{the:plmfn}) then the size of the vector grows exponentially to the memory growth. Therefore tighter upper bounds are obtained by Theorem \ref{the:ppm_upper} for functions $f(n) = \Omega(\log n)$ and by Theorem \ref{the:plmfn} for $f(n) = o(\log n)$. Note that for $f(n) = \log n$ the bounds by both theorems are the same.

The next theorem shows that for space functions $f(n) = \Omega(\log n)$, the PM model can simulate a nondeterministic TM of space $\mathcal{O}(nf(n))$ using $\mathcal{O}(f(n))$ in each agent. The intuition here is that with at least $\log n$ memory in each agent we can always assign unique ids and propagate the size of the population, as shown in Section \ref{sec:uids}, and therefore we can always organize the population into a nondeterministic TM where each of the $n$ agents provides $\mathcal{O}(f(n))$ space for the simulation.

\begin{theorem} \label{the:ppm_lower}
For any $f(n)=\Omega(\log n)$ it holds that $SNSPACE(nf(n))$ is a subset $PMSPACE(f(n))$.
\end{theorem}
\begin{proof}
In Theorem \ref{the:lowPLM} a nondeterministic TM of $\co(n \log n)$ space is simulated by a PALOMA protocol provided that all agents know the population size and have unique ids. In addition in Section \ref{sec:uids} a construction is provided that allows the PALOMA protocols to assume ids and population size knowledge. The same construction can be used for the PM protocols which use $\mathcal{O}(f(n))$ space for $f(n)=\Omega(\log n)$. The reason is that in this space every agent can store a unique id as well as the population size. Therefore the same protocol $\mathcal{I}$ of Theorem \ref{the:iplm} can be used here.
\qed
\end{proof}

From the previous two theorems we get a generalization of the exact characterization for all $PMSPACE($ $f(n))$, when $f(n) = \Omega(\log n)$. It is formally stated as:

\begin{theorem} \label{the:wlogn}
For any $f(n)=\Omega(\log n)$ it holds that $PMSPACE(f(n))=SN\-SPACE(nf(n))$.
\end{theorem}
\begin{proof}
It holds from theorems \ref{the:ppm_upper} and \ref{the:ppm_lower}.
\qed
\end{proof}

Thus, by considering together Corollary \ref{cor:ologn} and Theorem \ref{the:wlogn} we obtain the following corollary.
\begin{corollary}\label{cor:pm_thres}
$f(n)=\Theta(\log n)$ acts as a threshold for the computational power of the PM model.
\end{corollary}

Considering the above analysis and Corollary \ref{cor:pm_thres} it is worth noting that the PALOMA protocols, seem to belong to a golden section between realistic requirements, w.r.t. implementation, and computational power offered.

\begin{theorem}[Space Hierarchy Theorem of the PM model] \label{the:pmsh}
For any two functions $f,g:\mathbb{N} \rightarrow \mathbb{N}$, where $f(n) = \Omega(\log n)$ and $g(n) = o(f(n))$ there is a predicate $p$ in $PMSPACE(f(n))$  but not in $PMSPACE(g(n))$.
\end{theorem}
\begin{proof}
From Theorem \ref{the:nsymsh} we have that for any such functions $f,g$ there is a language $L \in SNSPACE(nf(n))$ so that $L \not\in SNSPACE(ng(n))$. From Theorem \ref{the:wlogn} we have that $PMSPACE(f(n))= SNSPACE(nf(n))$, therefore $p_L \in PMSPACE(f(n))$ (where $p_L$ is the symmetric predicate that corresponds to the symmetric language $L$). We distinguish two cases. If $g(n) = \Omega(\log n)$ then from Theorem \ref{the:wlogn} we have that $PMSPACE(g(n))= SNSPACE(ng(n))$ and so $L \not\in PMSPACE(g(n))$ or equivalently $p_L \not\in PMSPACE($ $g(n))$. If $g(n) = o(\log n)$ then from Corollary \ref{cor:ologn} we have that $PMSPACE(g(n)) \subsetneq SNSPACE(ng(n))$ $\subsetneq SNSPACE(nf(n)) = PMSPACE(f(n))$.
\qed
\end{proof}

\section{A threshold in the Computability of the PM Model}\label{threshold}

In this section we explore the computability of the PM model when the protocols use $o(\log \log n)$ space. We show that this bound acts as a threshold constraining the corresponding protocols to compute only semilinear predicates.

Let $SEMILINEAR$ denote the class of the semilinear predicates. We will prove that $SEMILINEAR=PMSPACE(f(n))$ when $f(n)=o(\log\log n)$. Moreover, we will prove that this bound is tight, so that $SEMILINEAR \subsetneq PMSPACE(f(n))$ when $f(n)=\mathcal{O}(\log \log n)$.

\subsection{$o(\log\log n)$ Threshold}

\begin{definition}
Let $A$ be a PM protocol executed in a population $V$ of size $n$. Define \emph{agent configuration graph}, $R_{A,V}=\{U,W,F\}$, a graph such that:
\begin{itemize}
 \item $U$ is the set of the agent configurations that can occur in any execution of $A$ such that the working flag is set to 0.
 \item $W$ is the set of edges $(u,v),~u,v\in U$ so that an agent configuration $v$ can occur from a configuration $u$ via a single interaction. 
\item $F:W\rightarrow U\times \{i,r\}$ so that when an agent $k$ being in configuration $u$ enters configuration $v$ via a single interaction with an agent being in configuration $w$, and $k$ acts as $x\in \{i,r\}$ (initiator-responder) in the interaction, then $F((u,v))=\{w,x\}$, while $F((u,v))=\emptyset$ in any other case. $F$ is a function since we only deal with deterministic protocols.
\end{itemize}
\end{definition}

In other words, $U$ contains the configurations that an agent may enter when we don't take into consideration the ones that correspond to internal computation, while $W$ defines the transitions between those configuration through interactions defined by $F$. Notice that in general, $R_{A,V}$ depends not only on the protocol $A$, but on the population $V$. We call a $u \in U$ \emph{initial node} iff it corresponds to an initial agent configuration.

Because of the uniformity property, we can deduce the following theorem:

\begin{theorem}\label{the:subgraph}
Let $R_{A,V}$, $R_{A,V^{\prime}}$ be two agent configuration graphs corresponding in a protocol $A$ for two different populations $V$, $V^{\prime}$ of size $n$ and $n^{\prime}$ respectively, with $n<n^{\prime}$. Then, there exists a subgraph $R^*$ of $R_{A,V^{\prime}}$ such that $R^* = R_{A,V}$, which contains all the initial nodes of $R_{A,V^{\prime}}$
\end{theorem}

\begin{proof}
Indeed, let $V^{\prime}_1$, $V^{\prime}_2$ be a partitioning of $V^{\prime}$ such that $V^{\prime}_1=V$, and observe the agent configuration graph that yielded by the execution of $A$ in $V^{\prime}_1$. Since both populations execute the same protocol $A$, and $V^{\prime}_1=V$, $R_{A,V}=R_{A,V^{\prime}_1}$. Moreover, since the initial nodes are the same for both populations, they must be in $R_{A,V^{\prime}_1}$. Finally, $R_{A,V^{\prime}_1}$ is a subgraph of $R_{A,V^{\prime}}$, as $V^{\prime}_1 \subset V^{\prime}$, and the proof is complete.
\qed
\end{proof}

The above theorem states that while we explore populations of greater size, the corresponding agent configuration graphs are only enhnanced with new nodes and edges, while the old ones are preserved.

Given an agent configuration graph, we associate each node $a$ with a value $r(a)$ inductively, as follows:

\begin{description}
  \item[Base Case] For any initial node $a$, $r(a)=r_{init}=1$.
  \item [Inductive Step] For any other node $a$, $r(a)=min(r(b)+r(c))$ such that $a$ is reachable from $b$ through an edge labeled as $c$, and $b$, $c$ have already been assigned an $r$ value.
\end{description}

\begin{lemma}\label{lem:rval}
Let $R_{A,V}=\{U,W,F\}$ be an agent configuration graph. Every node in $R_{A,V}$ get associated with an $r$ value.
\end{lemma}

\begin{proof}
Assume for the shake of the contradiction that there is a maximum, non empty set of nodes $U^{\prime} \subset U$ such that $\forall v\in U^{\prime}$, $v$ does not get associated with an $r$ value. Then $B=U-U^{\prime}$, and $C=(B,U^{\prime})$ defines a cut, with all the initial nodes being in $B$. We examine any edge $(u,v)$ labeled as $(w,x)$ that crosses the cut. Obviously $u\in B$ and $v\in U^{\prime}$, and $u$ is associated with a value $r(u)$. Since $v$ is not associated with any $r$ value, the same must hold for node $w$ (otherwise $r(v)=r(u)+r(w)$). We now examine the first agent $c$ that enters a configuration corresponding to some $v\in U^{\prime}$. Because of the above observation, this could only happen through an interaction with an agent being in a configuration that is also in $U^{\prime}$ which creates the contradiction.
\qed
\end{proof}

Note that for any given protocol and population size, the $r$ values are \emph{unique} since the agent configuration graph is unique. The following lemma captures a bound in the $r$ values when the corresponding protocol uses $f(n)=o(\log \log n)$ space.

\begin{lemma}\label{lem:rmax}
Let $r_{max-i}$ be the $i-$th greatest $r$ value associated with any node in an agent configuration graph. For any protocol $A$ that uses $f(n)=o(\log \log n)$, there exists a $n_0$ such that for any population of size $n>n_0$ $r_{max}<\frac{n}{2}$.
\end{lemma}

\begin{proof}
Since $f(n)=o(\log \log n)$, $\lim_{n \to \infty}\frac{f(n)}{\log \log n}=0$, so $\lim_{n \to \infty}\frac{\log \log n}{f(n)} = \infty$ and $\lim_{n \to \infty}\frac{\log n}{2^{f(n)}}=\infty$. It follows from the last equation that for any positive integer $M$, there exists a fixed $n_0$ such that $\frac{\log n}{2^{f(n)}}>M$ for any $n> n_0$.

Fix any such $n$ and let $k=|U|\leq 2^{f(n)}$ in the corresponding agent configuration graph. Since any node is associated with an $r$ value, there can be at most $k$ different such values. Now observe that $r_{max}\leq 2\cdot r_{max-1} \leq \cdots \leq 2^k\cdot r_{init}\leq 2^{2^f(n)}< 2^{\frac{\log n}{M}}\leq \sqrt[M]{n} \leq \frac{n}{2}$ for $n>max(n_0,2)$ and $M\geq 2$.
\qed
\end{proof}

The following lemma proves that the $r$ values correspond to some reachability properties in the agent configuration graph.
\begin{lemma}\label{lem:Qprop}
Let $Q(a)$ be the following property: \emph{Given a node $a$ in an agent configuration graph, there exists a population of size $r(a)$ and a fair execution of the corresponding protocol that will lead an agent to the configuration $a$}. $Q(a)$ holds for any node of the agent configuration graph.

\end{lemma}
\begin{proof}
We prove the above lemma by generalized induction in the $r$ values.
\begin{description}
  \item[Base Case:] $Q(a)$ obviously holds for any initial node $u$, since $r_{init}=1$.
  \item[Inductive Step:] We examine any non initial node $u$ that has been associated with a value $r(u)=r(a)+r(b)$, for some $a$, $b$. The inductive hypothesis guarantees that $Q(a)$ and $Q(b)$ hold. Then, a population of size $r(a)+r(b)$ can lead two agents to configurations $a$ and $b$ independently. Then an interaction between those agents will take one of them to configuration $u$, so $Q(u)$ holds too.
\end{description}
\qed
\end{proof}

Lemmata \ref{lem:rmax} and \ref{lem:Qprop} lead to the following:
\begin{lemma}\label{lem:inter}
For any protocol $A$ that uses $f(n)=o(\log \log n)$ there exists a fixed $n_0$ such that for any population of size $n>n_0$ and any pair of agent configurations $u,v$, the interaction $(u,v)$ can occur.
\end{lemma}

\begin{proof}
Indeed, because of the lemma \ref{lem:rmax}, there exists a $n_0$ such that for any $n>n_0$, $r(a)< \frac{n}{2}$ for any $a$. With that in mind, lemma \ref{lem:Qprop} guarantees that in any such population any interaction $(u,v)$ since any of the agent configurations $u$, $v$ can occur independently, by partitioning the population in two subpopulations of size $\frac{n}{2}$ each.
\qed
\end{proof}

We can complete our proof with the following theorem:

\begin{theorem}\label{the:semil}
Any PM protocol $A$ that uses $f(n)=o(\log \log n)$ can only compute semilinear predicates.
\end{theorem}

\begin{proof}
Because of the uniformity constraint, $A$ can be executed in any population of arbitrary size. We choose a fixed $n_0$ as defined in lemma \ref{lem:rmax} and examine any population $L$ of size $n>n_0$. Let $R_{A,L}$ be the corresponding agent configuration graph. Let $L^{\prime}$ be any population of size $n^{\prime}>n$ and $R_{A,L^{\prime}}$ the corresponding agent configuration graph. Because of the theorem \ref{the:subgraph}, $R_{A,L^{\prime}}$ contains a subgraph $K$, such that $K=R_{A,L}$, and the initial nodes of $R_{A,L^{\prime}}$ are in $K$. Let $U^*= U^{\prime}-U$, and $k$ the first agent configuration that appears in $L^{\prime}$ such that $k\in U^*$ through an interaction $(u,v)$($k$ can't be an initial configuration, thus it occurs through some interaction). Then $u,v\in U$, and the interaction $(u,v)$ can occur in the population $L$ too (lemma \ref{lem:inter}), so that $k\in U$, which refutes our choice of $k$ creating a contradiction. So, $U^*=\emptyset$, and the set of agent configurations does not change as we examine populations of greater size. Since the set of agent configurations remains unchanged, the corresponding predicate can be computed by the Population Protocol model, thus it is semilinear.
\qed
\end{proof}

In practice, the above theorem guarantees that for any protocol that uses only $f(n)=o(\log \log n)$ space in each agent, there exists a population of size $n_0$ in which it stops using extra space. Since $n_0$ is fixed, we can construct a protocol based on the agent configuration graph which uses constant space\footnote{Notice that the agent configuration graph can be viewed as a deterministic finite automaton.}, and thus can be executed in the Population Protocols Model.

So far, we have stablished that when $f(n)=o(\log \log n)$, $PMSPACE(f(n))\subseteq SEMILINEAR$. Since the inverse direction holds trivially, we can conclude that $PMSPACE(f(n)) = SEMILINEAR$.

\subsection{The Logarithmic Predicate}

We will now present the non semilinear logarithmic predicate, and devise a PM protocol that computes it using $O(\log\log n)$ space in each agent.

We define the logarithmic predicate as follows: During the initialization, each agent receives an input symbol from $X=\{a,0\}$, and let $N_a$ denote the number of agents that have received the symbol $a$. We want to compute whether $\log N_a = t$ for some arbitrary $t$. We give a high level protocol that computes this predicate, and prove that it can be correctly executed using $O(\log\log n)$ space.

Each agent $u$ maintains a variable $x_u$, and let $out_u$ be the variable that uses to write its output. Initially, any agent $u$ that receives $a$ as his input symbol sets $x_u=1$ and $out_u=1$, while any other agent $v$ sets $x_v=0$ and $out_v=1$.

Our main protocol consists of two subprotocols, $A$ and $B$, that are executed concurrently. Protocol $A$ does the following: whenever an interaction occurs between two agents, $u$, $v$, with $u$ being the initiator, if $x_u=x_v>0$, then $x_u=x_u+1$ and $x_v=0$. Otherwise, nothing happens. Protocol $B$ runs in parallel, and computes the semilinear predicate of determining whether there exist two or more agents having $x>0$. If so, it outputs $0$, otherwise it outputs $1$. Observe that $B$ is executed on \emph{stasbilizing inputs}, as the $x-$variables fluctate before they stabilize to their final value. However, it is well known that the semilinear predicates are also computable under this constraint (\cite{AAE06}).

\begin{lemma}
The main protocol uses $\mathcal{O}(\log \log n)$ space.
\end{lemma}

\begin{proof}
As protocol $B$ computes a semilinear predicate, it only uses $O(c)$ space, with $c$ being a constant. To examine the space bounds of $A$, pick any agent $u$. We examine the greatest value that can be  assigned to the variable $x_u$. Observe that in order for $x_u$ to reach value $k$, there have to be at least $2$ pre-existsing $x-$ variables with values $k-1$. Through an easy induction, it follows that there have to be at least $2^k$ pre-existing variables with the value $1$. Since $2^k\leq N_a$, $k\leq \log N_a \leq \log n$, so $x_u$ is upper bounded by $\log n$, thus consuming $\mathcal{O}(\log \log n)$ space.
\qed
\end{proof}

\begin{lemma}{\label{lem:log_com}}
For every agent $u$, eventually $out_u=1$ if $\log N_a=t$ for some arbitrary $t$, and $out_u=0$ otherwise.
\end{lemma}

\begin{proof}
Indeed, the execution of protocol $B$ guarantees that all agents will set $out=1$ iff eventually there exists only one agent $u$ that has a non-zero value assigned in $x_u$. Assume that $x_u=k$ for some $k$. Then, because of the analysis of lemma \ref{lem:log_com} during the initialization of the population will exist $2^k$ $x-$variables set to $1$. Since each of those variables corresponds to one $a$ assignement, $N_a=2^k\Rightarrow \log N_a = k$. On the other hand, if the answer of the protocol is $0$ then it means that there are $t > 1$ agents in the population with $x-$variables set to different values $x_1, x_2, \cdots, x_t$ otherwise they could have effective interactions with each other. Therefore, there should have initially existed $2^{x_1 - 1} + 2^{x_2 - 1} + \cdots + 2^{x_t - 1}$ agents with input $a$. This, however, means that $N_a \neq 2^k$ for any $k$ since each number can be uniquely expressed as a sum of powers of $2$. Thus the protocol correctly outputs $0$. 
\qed
\end{proof}

Thus, we have presented a non-semilinear predicate that can be computed by a PM protocol using $\mathcal{O}(\log \log n)$ space.
Combining this result with the one presente in the previous subsection, we obtain the following theorem:

\begin{theorem}[\emph{Threshold Theorem}]\label{the:logsemi}
 $SEMILINEAR = PMSPACE(f(n))$ when $f(n)=o(\log \log n)$. Moreover, $SEMILINEAR \subsetneq PMSPACE(f(n))$ when $f(n)=\mathcal{O}(\log \log n)$.
\end{theorem}
\section{Conclusions - Future Research Directions} \label{sec:conc}

We proposed the PM model, an extension of the PP model \cite{AADFP06}, in which the agents are communicating TMs. We mainly focused on studying the computational power of the new model when the space used by each agent is bounded by a logarithm of the population size. Although the model preserves uniformity and anonymity, interestingly, we have been able to prove that the agents can \emph{organize themselves into a nondeterministic TM} that makes full use of the agents' total memory (i.e. of $\mathcal{O}(n\log n)$ space). The agents are initially identical and have no global knowledge of the system, but by executing an \emph{iterative reinitiation process} they are able to assign \emph{unique consecutive ids} to themselves and get informed of the population size. In this manner, we showed that $PLM$, which is the class of predicates stably computable by the PM model using $\mathcal{O}(\log n)$ memory, contains all symmetric predicates in $NSPACE(n\log n)$. Moreover, by proving that $PLM$ is a subset of $NSPACE(n\log n)$, we concluded that it is precisely equal to the class consisting of all symmetric predicates in $NSPACE(n\log n)$. We then generalized this to $SNSPACE(nf(n))$ when the agents use $\mathcal{O}(f(n))$ memory, for all space functions $f(n)=\Omega(\log n)$. We also explored the behavior of the PM model for space bounds $f(n)=o(\log n)$ and proved that $SEMILINEAR=PMSPACE(f(n))$ when $f(n)=o(\log\log n)$. Finally, we showed that this bound is tight, that is, $SEMILINEAR \subsetneq PMSPACE(f(n))$ when $f(n)=\mathcal{O}(\log \log n)$.

Many interesting questions remain open. Is the PM model \emph{fault-tolerant}? What preconditions are needed in order to achieve satisfactory fault-tolerance? What is the behavior of the model when the agents use $\mathcal{O}(f(n))$ memory, where $f(n) = o(\log n)$ and $f(n) = \Omega(\log\log n)$.

\newpage

\section*{Acknowledgements}

We wish to thank some anonymous reviewers for their very useful comments to a previous version of this work.





%
%
%
%

\end{document}